\documentclass[a4paper,11pt]{article}
\usepackage[left=3.17cm,right=3.17cm,top=2.54cm,
headheight=0.5cm,headsep=0.54cm,bottom=2.54cm,footskip=0.79cm
]{geometry}
\usepackage[english]{babel}
\usepackage{amsmath,amssymb,amsthm}
\usepackage[all]{xy}
\usepackage{epsfig}
\usepackage{epstopdf}
\usepackage [latin1]{inputenc}

\newtheorem{theorem}{Theorem}
\newtheorem{lemma}{Lemma}

\newtheorem{example}{Example}

\usepackage{tikz}
\usetikzlibrary{arrows}
\usetikzlibrary{decorations.markings}

\usepackage[pdfstartview=FitH
]{hyperref}

\newcommand{\N}{\mathcal{N}}
\begin{document}

\title{Tighter Monogamy Relations in Multi-Qubit Systems}

\author{Yudie Gu$^{1}$, Yanmin Yang$^{1}$ $\thanks{e-mail: ym.yang@kust.edu.cn}$, Jialing Zhang$^{1}$, Wei Chen$^{2}$
\\[.3cm]
{\small $^{1}$ Faculty of Science,   Kunming University of Science and Technology,} \\
{\small Kunming,  650500, China}\\[.3cm]
{\small $^{2}$ School of Computer Science and Technology, Dongguan University of Technology,} \\
{\small Dongguan, 523808, China}\\[.3cm]
}

\maketitle

\begin{abstract}
 In this paper, we present some monogamy relations of multiqubit quantum entanglement in terms of the $\beta$th power of concurrence, entanglement of formation and convex-roof extended  negativity. These monogamy relations are proved to be tighter than the existing ones, together with detailed examples showing the tightness.
\end{abstract}


\bigskip

\section{Introduction}
Quantum entanglement is widely used as a very important resource in quantum information processing \cite{Jafarpour,Deng,Huang,Wang}. With the emergence of quantum information theory, quantum entanglement plays a very important role in quantum cryptography, quantum teleportation and measurement based quantum computing. An important issue related to the entanglement metric is the limited shareability of the two-part entanglement in a multipartite entangled qubit system, that is, the single duality of entanglement \cite{Terhal}. Monogamy of entanglement (MoE) plays a very important role in many quantum information and communication processing tasks, such as security proof of quantum cryptography schemes and security analysis of quantum key distribution \cite{Bennett,Pawlowski}.

For a tripartite quantum state $\rho_{ABC}$, MoE can be described as $E(\rho_{A|BC})\geq E(\rho_{AB})+E(\rho_{AC})$, where $\rho_{AB}={\rm tr}_C(\rho_{ABC})$, $\rho_{AC}={\rm tr}_B(\rho_{ABC})$, $E(\rho_{A|BC})$ denotes the entanglement between systems A and BC. A remarkable result was established by Coffman, Kundu and Wootters (CKW) \cite{Coffman} for three qubits, that was the simultaneous squares satisfy monogamy inequality. Then, the so-called CKW inequality was generalized to any $N$-qubit system \cite{Osborne}. Interestingly, it is further proved that similar inequalities of polyqubit monogamy can be established for negativity and convex-roof extended negativity (CREN) \cite{Kim,Ou,Yang}, the entanglement of formation (EoF) \cite{Rungta,Ren}, R\'enyi-$\alpha$ entanglement \cite{Sanders,Vedral} and Tsallis-$q$ entanglement \cite{Yu}.


Our paper is organized as follows. In Sec.\ref{sec2}, we present and prove  two  monogamy inequalities  for the $\beta$th ($ \beta\geq 2$) power of concurrence in $N$-qubit system.
In Sec.\ref{sec3}, we give a  tighter monogamy relation for the $\beta$th ($ \beta\geq \sqrt{2}$) power of EoF in $2\otimes2\otimes2^{N-2}$ system. Then we  extend the result to $N$-qubit system.
In Sec.\ref{sec4}, the monogamy relation for the $\beta$th ($ \beta\geq 2$) power of CREN  in $N$-qubit system is discussed. In addition,  detailed examples are  given to illustrate the tightness.
In Sec.\ref{sec5}, we summarize our results.

\section{Tighter Monogamy Relations using Concurrence}\label{sec2}
Given a bipartite pure state $|\phi\rangle_{AB}$ on Hilbert space ${ H_A\otimes  H_B}$, the concurrence is given by \cite{Uhlmann,Albeverio,Rungta1}
\begin{eqnarray}\label{def of C}
C(|\phi\rangle_{AB})=\sqrt{2(1-{\rm{Tr}}(\rho^2_A))}
\label{eqn01},
\end{eqnarray}
where $\rho_{A}$ is the reduced density matrix by tracing over the subsystem $B$, $\rho_A={\rm{Tr}}_B(|\phi\rangle_{AB}\langle\phi|)$. For a bipartite mixed state $\rho_{AB}$, the concurrence is defined by the convex-roof,
\begin{eqnarray}
C(\rho_{AB})=\min_{\{p_i,|\phi_i\rangle\}}
\sum_ip_iC(|\phi_i\rangle_{AB})
\label{eq3},
\end{eqnarray}
where the minimum is taken over all possible pure state decompositions of $\rho_{AB}=\sum_ip_i|\phi_i\rangle\langle\phi_i|$, with $\sum_ip_i=1$ and $p_i\geq 0$.

For any N-qubit mixed state $\rho_{AB_1\cdots B_{N-1}}$, the concurrence $C(\rho_{A|B_1\cdots B_{N-1}})$ of the state $\rho_{AB_1\cdots B_{N-1}}$ under bipartite partition $A$ and $B_1\cdots B_{N-1}$ satisfies \cite{Zhu}
\begin{equation}\label{Con1}
\begin{array}{rl}
&C^\beta(\rho_{A|B_1\cdots B_{N-1}})\ \ \geq C^\beta(\rho_{AB_1})+C^\beta(\rho_{AB_2})+\cdots+C^\beta(\rho_{AB_{N-1}}),
\end{array}
\end{equation}
for $\beta\geq2$. Furthermore, for an N-qubit mixed state, if $C_{AB_i}\geq C_{A|B_{i+1}\cdots B_{N-1}}$ for $i=1,2,\ldots, m$, and $C_{AB_j}\leq C_{A|B_{j+1}\cdots B_{N-1}}$ for $j={m+1},\ldots, {N-2}$, a generalized monogamy relation for $\beta\geq2$ was presented as \cite{Jin}
\begin{equation}\label{Con2}
\begin{array}{rl}
&C^\beta(\rho_{A|B_1\cdots B_{N-1}})\\
&\ \geq C^\beta(\rho_{AB_1})+(2^{\frac{\beta}{2}}-1)C^\beta(\rho_{AB_2})+\cdots+\big(2^{\frac{\beta}{2}}-1)^{m-1}C^\beta(\rho_{AB_{m}})\\
&\ \ \ +\big(2^{\frac{\beta}{2}}-1)^{m+1}[C^\beta(\rho_{AB_{m+1}})+\cdots+ C^\beta(\rho_{AB_{N-2}})]+\big(2^{\frac{\beta}{2}}-1)^m C^\beta(\rho_{AB_{N-1}}),
\end{array}
\end{equation}
where $ 1\leq m\leq N-3$, $N\geq4$.

In the following, we will show that these monogamy  relations for concurrence can be further tightened under some conditions. Before that, we first introduce two lemmas as follows.

\begin{lemma}\label{lemma1}
For any $x\in[0,1]$ and $t\geq1$, we have
\begin{equation}\label{lemma1-equality}
\begin{array}{l}
(1+x)^t\geq1+\frac{t}{2}x+\frac{(t-1)^2}{4}x^2+(2^t-\frac{t}{2}+\frac{(t-1)^2}{4}-1)x^t-{\frac{(t-1)^2}{2}}x^{t+1}\\
\geq1+\frac{t}{2}x+(2^t-\frac{t}{2}-1)x^t\geq1+(2^t-1)x^t.
\end{array}
\end{equation}
\end{lemma}

\begin{proof}
Let us consider the function $f(t,x)=\frac{(1+x)^t-1-\frac{t}{2}x-\frac{(t-1)^2}{4}x^2+\frac{(t-1)^2}{2}x^{t+1}}{x^t}$. Then, $\frac{\partial f(t,x)}{\partial x}=\frac{tx^{t-1}[1+\frac{t-1}{2}x+\frac{(t-1)^2(t-2)}{4t}x^2+\frac{(t-1)^2}{2t}x^{t+1}-(1+x)^{t-1}]}{x^{2t}}$.
Next, we will prove that
\begin{equation}\label{lem1-equality1}
1+\frac{t-1}{2}x+\frac{(t-1)^2(t-2)}{4t}x^2+\frac{(t-1)^2}{2t}x^{t+1}\leq(1+x)^{t-1},
\end{equation}
thus $\frac{\partial f(t,x)}{\partial x}\leq0$, $f(t,x)$ is a decreasing function of $x$, i.e. $f(t,x)\geq f(t,1)=2^t-\frac{t}{2}+\frac{(t-1)^2}{4}-1$. It follows that $(1+x)^t\geq1+\frac{t}{2}x+\frac{(t-1)^2}{4}x^2+(2^t-\frac{t}{2}+\frac{(t-1)^2}{4}-1)x^t-\frac{(t-1)^2}{2}x^{t+1}$.

For the case $1\leq t\leq 2$,  it is obvious that $(1+x)^{t-1}\geq 1+(t-1)x+\frac{(t-1)(t-2)}{2}x^2$. Besides, we have
\begin{equation*}
  \begin{array}{rcl}
\frac{(t-1)(t-2)}{2}x^2 & =& \frac{t-1}{4t} 2t(t-2)x^{2}  =  \frac{t-1}{4t}[(t-1)(t-2)x^{2}+(t^2+t-2)x^{2}-2tx^{2}]\\
 & \geq & \frac{t-1}{4t}[(t-1)(t-2)x^{2}+(2t-2)x^{t+1}-2tx].\\
  \end{array}
\end{equation*}
Thus, the Eq. (\ref{lem1-equality1}) is hold.

For the case  $t\geq2$, it is obvious that $(1+x)^{t-1}\geq 1+(t-1)x+\frac{(t-1)(t-2)}{4}x^2$. Besides, we have
\begin{equation*}
  \begin{array}{rcl}
\frac{(t-1)(t-2)}{4}x^2 & =& \frac{t-1}{4t} t(t-2)x^{2} =  \frac{t-1}{4t}[(t-1)(t-2)x^{2}+(2t-2)x^{2}-tx^2]\\
 & \geq & \frac{t-1}{4t}[(t-1)(t-2)x^{2}+2(t-1)x^{t+1}-2tx].\\
  \end{array}
\end{equation*}
Thus, the Eq. (\ref{lem1-equality1}) is hold.

On the other hand, Since $x^2-2x^{t+1}+x^t\geq0$ and $\frac{(t-1)^2}{4}\geq0$, for $t\geq1$ and $x\in[0,1]$, we can get $(1+x)^t\geq1+\frac{t}{2}x+\frac{(t-1)^2}{4}x^2+(2^t-\frac{t}{2}+\frac{(t-1)^2}{4}-1)x^t-\frac{(t-1)^2}{2}x^{t+1}
\geq1+\frac{t}{2}x+(2^t-\frac{t}{2}-1)x^t\geq1+(2^t-1)x^t$.
\end{proof}

\begin{lemma} \label{lemma2}
For any mixed state $\rho_{ABC}$ in a $2\otimes2\otimes2^{N-2}$ system, suppose that $C_{AB}\geq C_{AC}$, we have
\begin{equation}\label{lemma2-equality}
C_{A|BC}^\beta\geq C_{AB}^\beta+hC_{AC}^\beta+\frac{\beta}{4}C_{AC}^2(C_{AB}^{\beta-2}-C_{AC}^{\beta-2})+\frac{(\beta-2)^2}{16}C_{AC}^4(C_{AB}^{\beta-4}+C_{AC}^{\beta-4}-2C_{AC}^{\beta-2}C_{AB}^{-2}),
\end{equation}
for all $\beta\geq2$, where $h=2^{\frac{\beta}{2}}-1$, $C_{A|BC}=C(\rho_{A|BC})$, analogously for $C_{AB}$ and $C_{AC}$.
\end{lemma}
\begin{proof}
Since $C_{AB}\geq C_{AC}$, we obtain
\begin{equation}
\begin{array}{l}
C_{A|BC}^\beta  \geq (C_{AB}^2+C_{AC}^2)^{\frac{\beta}{2}}=  C_{AB}^\beta \Big(1+\frac{C_{AC}^2}{C_{AB}^2}\Big)^{\frac{\beta}{2}}\\
 \geq  C_{AB}^\beta \Big[1+\frac{\beta}{4}\frac{C_{AC}^2}{C_{AB}^2}+\frac{(\beta-2)^2}{16}\frac{C_{AC}^4}{C_{AB}^4}+(2^{\frac{\beta}{2}}-\frac{\beta}{4}
+\frac{(\beta-2)^2}{16}-1)\frac{C_{AC}^\beta}{C_{AB}^\beta}-\frac{(\beta-2)^2}{8}\frac{C_{AC}^{\beta+2}}{C_{AB}^{\beta+2}}\Big]\\
 = C_{AB}^\beta+hC_{AC}^\beta+\frac{\beta}{4}C_{AC}^2(C_{AB}^{\beta-2}-C_{AC}^{\beta-2})+\frac{(\beta-2)^2}{16}C_{AC}^4(C_{AB}^{\beta-4}+C_{AC}^{\beta-4}
-2C_{AC}^{\beta-2}C_{AB}^{-2}),
\end{array}
\end{equation}
where the first inequality is due to the fact that $C_{A|BC}^2\geq C_{AB}^2+C_{AC}^2$ for
any $2\otimes2\otimes2^{N-2}$ tripartite state $\rho_{A|BC}$ \cite{Osborne,XJR} and the second inequality is due to Lemma \ref{lemma1}.
\end{proof}

\begin{theorem}\label{thm1}
For any $N$-qubit mixed state $\rho_{AB_1\cdots B_{N-1}}$,  if $C_{AB_i}\geq C_{A|B_{i+1}\cdots B_{N-1}}$, for $i=1,2,\ldots, N-2$, we have
\begin{equation}
\begin{array}{rl}
C_{A|B_1\cdots B_{N-1}}^\beta\geq\sum\limits_{i=1}^{N-2}h^{i-1}(C_{AB_i}^\beta+P_{AB_i})+h^{N-2} C_{AB_{N-1}}^\beta,
\end{array}
\end{equation}
for all $N\geq3$, $\beta\geq2$, where $h=2^{\frac{\beta}{2}}-1$, $P_{AB_i}=\frac{\beta}{4}C_{A|B_{i+1}\ldots B_{N-1}}^2(C_{AB_i}^{\beta-2}-C_{A|B_{i+1}\ldots B_{N-1}}^{\beta-2})+
\frac{(\beta-2)^2}{16}C_{A|B_{i+1}\ldots B_{N-1}}^4(C_{AB_i}^{\beta-4}+C_{A|B_{i+1}\ldots B_{N-1}}^{\beta-4}-2C_{A|B_{i+1}\ldots B_{N-1}}^{\beta-2}C_{AB_i}^{-2})$.
\end{theorem}
\begin{proof}
Due to  Eq. (\ref{lemma2-equality}), we obtain
\begin{equation}
\begin{array}{rl}
&C_{A|B_1\ldots B_{N-1}}^\beta \\
& \geq C_{AB_1}^\beta+hC_{A|B_2\ldots B_{N-1}}^\beta
+\frac{\beta}{4}C_{A|B_2\ldots B_{N-1}}^2(C_{AB_1}^{\beta-2}-C_{A|B_2\ldots B_{N-1}}^{\beta-2})\\
&\ \ \ +\frac{(\beta-2)^2}{16}C_{A|B_2\ldots B_{N-1}}^4(C_{AB_1}^{\beta-4}+C_{A|B_2\ldots B_{N-1}}^{\beta-4}-2C_{A|B_{2}\ldots B_{N-1}}^{\beta-2}C_{AB_1}^{-2})\\
& \geq C_{AB_1}^\beta+h\Big[C_{AB_2}^\beta+hC_{A|B_3\ldots B_{N-1}}^\beta+\frac{\beta}{4}C_{A|B_3\ldots B_{N-1}}^2(C_{AB_2}^{\beta-2}-C_{A|B_3\ldots B_{N-1}}^{\beta-2})\\
&\ \ \ +\frac{(\beta-2)^2}{16}C_{A|B_3\ldots B_{N-1}}^4(C_{AB_2}^{\beta-4}+C_{A|B_3\ldots B_{N-1}}^{\beta-4}-2C_{A|B_{3}\ldots B_{N-1}}^{\beta-2}C_{AB_2}^{-2})\Big]\\
&\ \ \ +\frac{\beta}{4}C_{A|B_2\ldots B_{N-1}}^2(C_{AB_1}^{\beta-2}-C_{A|B_2\ldots B_{N-1}}^{\beta-2})\\
&\ \ \ +\frac{(\beta-2)^2}{16}C_{A|B_2\ldots B_{N-1}}^4(C_{AB_1}^{\beta-4}+C_{A|B_2\ldots B_{N-1}}^{\beta-4}-2C_{A|B_{2}\ldots B_{N-1}}^{\beta-2}C_{AB_1}^{-2})\\
& \geq \cdots \\
& \geq C_{AB_1}^\beta+hC_{AB_2}^\beta+\cdots+h^{N-2}C_{AB_{N-1}}^\beta
+h^{N-3}\Big[\frac{\beta}{4}C_{AB_{N-1}}^2(C_{AB_{N-2}}^{\beta-2}-C_{AB_{N-1}}^{\beta-2})\\
&\ \ \ +\frac{(\beta-2)^2}{16}C_{AB_{N-1}}^4(C_{AB_{N-2}}^{\beta-4}+C_{AB_{N-1}}^{\beta-4}-2C_{AB_{N-1}}^{\beta-2}C_{AB_{N-2}}^{-2})\Big]\\
&\ \ \ +\cdots +h\Big[\frac{\beta}{4}C_{A|B_3\ldots B_{N-1}}^2(C_{AB_2}^{\beta-2}-C_{A|B_3\ldots B_{N-1}}^{\beta-2})\\
&\ \ \ +\frac{(\beta-2)^2}{16}C_{A|B_3\ldots B_{N-1}}^4(C_{AB_2}^{\beta-4}+C_{A|B_3\ldots B_{N-1}}^{\beta-4}-2C_{A|B_3\ldots B_{N-1}}^{\beta-2}C_{AB_{2}}^{-2})\Big]\\
&\ \ \ +\frac{\beta}{4}C_{A|B_2\ldots B_{N-1}}^2(C_{AB_1}^{\beta-2}-C_{A|B_2\ldots B_{N-1}}^{\beta-2})\\
&\ \ \ +\frac{(\beta-2)^2}{16}C_{A|B_2\ldots B_{N-1}}^4(C_{AB_1}^{\beta-4}+C_{A|B_2\ldots B_{N-1}}^{\beta-4}-2C_{A|B_2\ldots B_{N-1}}^{\beta-2}C_{AB_{1}}^{-2}).
\end{array}
\end{equation}
By  the denotation of $P_{AB_i}$, we complete the proof.
\end{proof}

\begin{theorem}
For any $N$-qubit mixed state $\rho_{AB_1\cdots B_{N-1}}$, if  $C_{AB_i}\geq C_{A|B_{i+1}\cdots B_{N-1}}$ for $i=1,2,\ldots,m$, and $C_{AB_j}\leq C_{A|B_{j+1}\cdots B_{N-1}}$ for $j=m+1,\ldots,N-2$, $\forall ~ 1\leq m\leq N-3$, we have
\begin{equation}\label{thm2-e}
\begin{array}{rl}
C_{A|B_1\cdots B_{N-1}}^\beta
\geq
\sum\limits_{i=1}^{m}h^{i-1}(C_{AB_i}^\beta+P_{AB_i})+
h^m\sum\limits_{j=m+1}^{N-2}(hC_{AB_j}^\beta+P_{AB_j}^1)+h^m C_{AB_{N-1}}^\beta,
\end{array}
\end{equation}
for all $N\geq4$,  $\beta\geq2$, where $h=2^{\frac{\beta}{2}}-1$, $P_{AB_i}=\frac{\beta}{4}C_{A|B_{i+1}\ldots B_{N-1}}^2(C_{AB_i}^{\beta-2}-C_{A|B_{i+1}\ldots B_{N-1}}^{\beta-2})+
\frac{(\beta-2)^2}{16}C_{A|B_{i+1}\ldots B_{N-1}}^4(C_{AB_i}^{\beta-4}+C_{A|B_{i+1}\ldots B_{N-1}}^{\beta-4}-2C_{A|B_{i+1}\ldots B_{N-1}}^{\beta-2}C_{AB_i}^{-2})$,\\
$P_{AB_j}^1=\frac{\beta}{4}C_{AB_j}^2(C_{A|B_{j+1}\ldots B_{N-1}}^{\beta-2}-C_{AB_j}^{\beta-2})
+\frac{(\beta-2)^2}{16}C_{AB_j}^4(C_{A|B_{j+1}\ldots B_{N-1}}^{\beta-4}+C_{AB_j}^{\beta-4}-2C_{AB_j}^{\beta-2}C_{A|B_{j+1}\ldots B_{N-1}}^{-2})$.
\end{theorem}

\begin{proof}
Due to the proof process of Theorem \ref{thm1}, we can get that
\begin{equation} \label{thm-e1}
C_{A|B_1\ldots B_{N-1}}^\beta\geq \sum\limits_{i=1}^{m}h^{i-1}(C_{AB_i}^\beta+P_{AB_i})+h^{m}C_{A|B_{m+1}\ldots B_{N-1}}^\beta.
\end{equation}
In addition, since $C_{AB_j}\leq C_{A|B_{j+1}\cdots B_{N-1}}$ for $j=m+1,\ldots,N-2$, hence,
\begin{equation}\label{thm-e2}
\begin{array}{rl}
&C_{A|B_{m+1}\ldots B_{N-1}}^\beta \\
& \geq C_{A|B_{m+2}\ldots B_{N-1}}^\beta+hC_{AB_{m+1}}^\beta
+\frac{\beta}{4}C_{AB_{m+1}}^2(C_{A|B_{m+2}\ldots B_{N-1}}^{\beta-2}-C_{AB_{m+1}}^{\beta-2})\\
&\ \ \ +\frac{(\beta-2)^2}{16}C_{AB_{m+1}}^4(C_{A|B_{m+2}\ldots B_{N-1}}^{\beta-4}+C_{AB_{m+1}}^{\beta-4}-2C_{AB_{m+1}}^{\beta-2}C_{A|B_{m+2}\ldots B_{N-1}}^{-2})\\
& \geq \sum\limits_{j=m+1}^{N-2}(hC_{AB_j}^\beta+P_{AB_j}^1)+C_{AB_{N-1}}^\beta.
\end{array}
\end{equation}
Combing  Eq. (\ref{thm-e1}) and Eq. (\ref{thm-e2}), we can get the inequality (\ref{thm2-e}).
\end{proof}

\begin{example}\label{ex1}
 Consider the three-qubit state $|\psi\rangle_{ABC}$ in generalized Schmidt decomposition form \cite{Acin,SM}
\begin{equation}\label{Con6}
|\psi\rangle_{ABC}=\lambda_0|000\rangle+\lambda_1e^{i\varphi}|100\rangle+\lambda_2|101\rangle+\lambda_3|110\rangle+\lambda_4|111\rangle,
\end{equation}
where $\lambda_i\geq0$, $i=0,1,2,3,4$, and $\sum\limits_{i=0}^4\lambda_i^2=1$. A direct calculation shows that $C_{A|BC}=2\lambda_0\sqrt{\lambda_2^2+\lambda_3^2+\lambda_4^2}$, $C_{AB}=2\lambda_0\lambda_2$ and
$C_{AC}=2\lambda_0\lambda_3$.

Set $\lambda_0=\frac{\sqrt{2}}{3},\lambda_1=0,\lambda_2=\frac{\sqrt{5}}{3},\lambda_3=\frac{\sqrt{2}}{3},\lambda_4=0$.
We have $C_{A|BC}=\frac{2\sqrt{14}}{9}$, $C_{AB}=\frac{2\sqrt{10}}{9}$ and $C_{AC}=\frac{4}{9}$.
Then $C_{A|BC}^\beta=(\frac{2\sqrt{14}}{9})^\beta\geq C_{AB}^\beta+hC_{AC}^\beta+\frac{\beta}{4}C_{AC}^2(C_{AB}^{\beta-2}-C_{AC}^{\beta-2})+\frac{(\beta-2)^2}{16}C_{AC}^4(C_{AB}^{\beta-4}+C_{AC}^{\beta-4}
-2C_{AC}^{\beta-2}C_{AB}^{-2})
=(\frac{2\sqrt{10}}{9})^\beta+h(\frac{4}{9})^\beta+\frac{\beta}{4}(\frac{4}{9})^2\Big[(\frac{2\sqrt{10}}{9})^{\beta-2}-(\frac{4}{9})^{\beta-2}\Big]+\frac{(\beta-2)^2}{16}(\frac{4}{9})^4\Big[(\frac{2\sqrt{10}}{9})^{\beta-4}+(\frac{4}{9})^{\beta-4}-2(\frac{4}{9})^{\beta-2}(\frac{2\sqrt{10}}{9})^{-2}\Big]$.
 However, the result in \cite{JZSZ} is $C_{AB}^\beta+hC_{AC}^\beta+\frac{\beta}{4}C_{AC}^2(C_{AB}^{\beta-2}-C_{AC}^{\beta-2})=(\frac{2\sqrt{10}}{9})^\beta+h(\frac{4}{9})^\beta+\frac{\beta}{4}(\frac{4}{9})^2\Big[(\frac{2\sqrt{10}}{9})^{\beta-2}-(\frac{4}{9})^{\beta-2}\Big]$. We can see that our results are better than the ones in \cite{JZSZ} for $\beta\geq2$, see FIG. \ref{1}.
\begin{figure}
  \centering
  \includegraphics[width=0.75\textwidth]{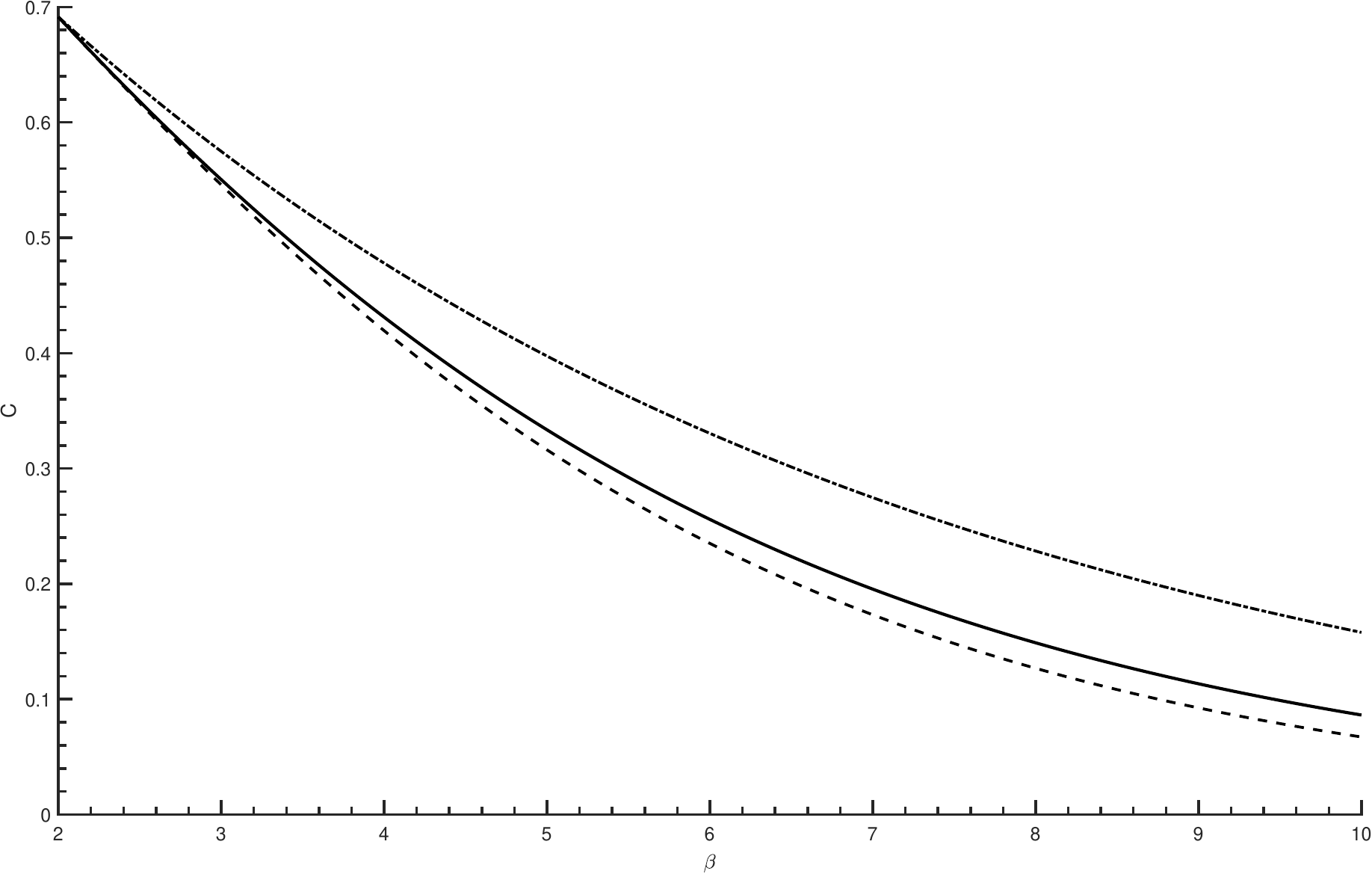}\\
  \caption{Dash dotted line ,   $C^\beta_{A|BC}$ as a function of $\beta$ ($2\leq\beta\leq10$); solid line,  the lower bound of $C^\beta_{A|BC}$ as a function of $\beta$ ($2\leq\beta\leq10$) in Eq. (\ref{thm2-e});  dash  line, the lower bound of $C^\beta_{A|BC}$ as a function of $\beta$ ($2\leq\beta\leq10$) in \cite{JZSZ}.}\label{1}
\end{figure}
\end{example}

\section{Tighter Monogamy Relations using EoF}\label{sec3}
Let $H_A$ and $H_B$ be two Hilbert spaces with dimension $m$ and $n$ $(m\leq n)$. The entanglement of formation (EoF) of a pure state $|\phi\rangle_{AB}$ on Hilbert space ${H_A\otimes  H_B}$,  is defined as \cite{Bennett1,CHB}
\begin{eqnarray} \label{def EoF pure}
E(|\phi\rangle_{AB})=S(\rho_A)=-{\rm{Tr}}(\rho_A\log_2\rho_A),
\end{eqnarray}
where $S(\rho)=-\rm{Tr}(\rho\log_2\rho)$ and $\rho_A={\rm Tr}_B(|\phi\rangle_{AB}\langle\phi|)$. For a bipartite mixed state $|\phi\rangle_{AB}$ on Hilbert space ${H_A\otimes  H_B}$, the EoF is given by
\begin{eqnarray}\label{def EoF mixed}
E(\rho_{AB})=\inf_{\{p_i,|\phi_i\rangle\}}\sum_ip_iE(|\phi_i\rangle),
\end{eqnarray}
where the infimum is taken over all possible pure state decompositions of $\rho_{AB}$.

Let $g(x)=H\big(\frac{1+\sqrt{1-x}}{2}\big)$ and $H(x)=-x\log_2x-(1-x)\log_2(1-x)$,
it is obvious that $g(x)$ is a monotonically increasing function for $0\leq x\leq1$, and satisfies
\begin{equation}\label{ine for g}
g^{\sqrt{2}}(x^2+y^2)\geq g^{\sqrt{2}}(x^2)+g^{\sqrt{2}}(y^2),
\end{equation}
where $g^{\sqrt{2}}(x^2+y^2)=[g(x^2+y^2)]^{\sqrt{2}}$.

From Eqs. (\ref{def EoF pure}) and (\ref{def EoF mixed}), we have $E(|\phi\rangle)=g(\mathcal{C}^2(|\phi\rangle))$ for $2\otimes d \ (d\geq2)$ pure state $|\phi\rangle$. And $E(\rho)=g(\mathcal{C}^2(\rho))$ for arbitrary two-qubit mixed state $\rho$ \cite{Wootters}.

Wootters \cite{Coffman} shows that the EoF dose not satisfy the monogamy inequality $E_{AB}+E_{AC}\leq E_{A|BC}$. In \cite{BZYW}, the authors shows that EoF is a monotonic function satisfying $E^2(C_{A|B_1B_2\cdots B_{N-1}}^2)\geq E^2 \sum_{i=1}^{N-1}(C_{AB_i}^2)$. For $N$-qubit systems, one has \cite{Zhu}
\begin{equation}\label{EoF1}
E^\beta_{A|B_1B_2\cdots B_{N-1}}\geq E^\beta_{AB_1}+E^\beta_{AB_2}+\cdots+E^\beta_{AB_{N-1}},
\end{equation}
for $\beta\geq\sqrt{2}$, where $E_{A|B_1B_2\cdots B_{N-1}}$ is the EoF of $\rho$ under bipartite partition $A|B_1B_2\cdots B_{N-1}$,
$E_{AB_i}$ is the EoF of the mixed state $\rho_{AB_i}={\rm Tr}_{B_1\cdots B_{i-1},B_{i+1}\cdots B_{N-1}}(\rho)$
for $i=1,2,\ldots,N-1$.

\begin{lemma}\label{lemma3}
For any mixed state $\rho_{ABC}$ in a $2\otimes2\otimes2^{N-2}$ system, $\beta\geq\sqrt{2}$, if $C_{AB}\geq C_{AC}$, then we have
\begin{equation}
\begin{array}{rl}
E_{A|BC}^\beta\geq & E_{AB}^\beta+hE_{AC}^\beta+\frac{t}{2}E_{AC}^{\sqrt{2}}(E_{AB}^{\beta-{\sqrt{2}}}-E_{AC}^{\beta-{\sqrt{2}}})\\
& +\frac{(t-1)^2}{4}E_{AC}^{2\sqrt{2}}(E_{AB}^{\beta-2\sqrt{2}}+E_{AC}^{\beta-2\sqrt{2}}-2E_{AC}^{\beta-\sqrt{2}}E_{AB}^{-\sqrt{2}}),
\end{array}
\end{equation}
where $t=\frac{\beta}{\sqrt{2}}, h=2^t-1$.
\end{lemma}

\begin{proof}
The proof is similar to  the proof of Lemma \ref{lemma2}.
\end{proof}

In fact, the result can be generalized to  $N$-qubit mixed state $\rho_{AB_1\cdots B_{N-1}}$. The following theorem holds for $\rho_{AB_1\cdots B_{N-1}}$.

\begin{theorem}
For any $N$-qubit mixed state $\rho_{AB_1\cdots B_{N-1}}$, if $C_{AB_i}\geq C_{A|B_{i+1}\cdots B_{N-1}}$ for $i=1,2,\ldots,N-2$,  we have
\begin{equation}\label{thm3-e}
\begin{array}{rl}
E_{A|B_1\cdots B_{N-1}}^\beta\geq\sum\limits_{i=1}^{N-2}h^{i-1}(E_{AB_i}^\beta+Q_{AB_i})+h^{N-2} E_{AB_{N-1}}^\beta,
\end{array}
\end{equation}
for $\beta\geq\sqrt{2}$, where $h=2^{t}-1$, $t=\frac{\beta}{\sqrt{2}}$, $Q_{AB_i}=\frac{t}{2}(E_{AB_{i+1}}^{\sqrt{2}}+\ldots+E_{AB_{N-1}}^{\sqrt{2}})(E_{AB_i}^{\beta-{\sqrt{2}}}-E_{A|{B_{i+1}\ldots B_{N-1}}}^{\beta-{\sqrt{2}}})+\frac{(t-1)^2}{4}(E_{AB_{i+1}}^{2\sqrt{2}}+\ldots+E_{AB_{N-1}}^{2\sqrt{2}})
[E_{AB_i}^{\beta-{2\sqrt{2}}}+\ldots+E_{AB_{N-1}}^{\beta-{2\sqrt{2}}}-2(E_{A|{B_{i+1}\ldots B_{N-1}}}^{\beta-\sqrt{2}})E_{AB_i}^{-\sqrt{2}}]$.
\end{theorem}

\begin{proof}
Let $\rho=\sum_ip_i|\psi_i\rangle\langle\psi_i|\in H_A\otimes H_{B_1}\otimes\cdots H_{B_{N-1}}$ be the optimal decomposition of
$E_{A|B_1B_2\cdots B_{N-1}}(\rho)$ for the $N$-qubit mixed state $\rho$, we have \cite{Jin}
\begin{equation}\label{thm3e1}
E_{A|B_1B_2\cdots B_{N-1}}\geq g(C^2_{A|B_1B_2\cdots B_{N-1}}).
\end{equation}

In addition, for $\beta\geq\sqrt{2}$ , we have
\begin{equation}\label{thm3e2}
\begin{array}{rl}
& g^{\beta}(x^2+y^2)\\
=& \Big[g^{\sqrt{2}}(x^2+y^2)\Big]^t\geq  \Big[g^{\sqrt{2}}(x^2)+g^{\sqrt{2}}(y^2)\Big]^t\\
\geq &    g^{\beta}(x^2)+(2^t-1)g^{\beta}(y^2)+\frac{t}{2}g^{\sqrt{2}}(y^2)\Big[g^{\beta-{\sqrt{2}}}(x^2)-g^{\beta-{\sqrt{2}}}(y^2)\Big]\\
&  +\frac{(t-1)^2}{4}g^{2\sqrt{2}}(y^2)\Big[g^{\beta-2\sqrt{2}}(x^2)+g^{\beta-2\sqrt{2}}(y^2)
-2g^{\beta-\sqrt{2}}(y^2)g^{-\sqrt{2}}(x^2)\Big],
\end{array}
\end{equation}
where the first inequality is due to Eq. (\ref{ine for g}), and without  loss of generality, we can assume  $x^2\geq y^2$, then the second inequality is obtained from the monotonicity of $g(x)$ and Eq. (\ref{lemma1-equality}).

Thus,  combining Eqs. (\ref{thm3e1}) and (\ref{thm3e2}), we obtain
\begin{equation}
\begin{array}{rl}
&E^\beta_{A|B_1B_2\cdots B_{N-1}}\\
&\ \ \geq g^\beta(C^2_{AB_1}+C^2_{AB_2}+\ldots +C^2_{AB_{N-1}})\\
&\ \ \geq g^\beta(C^2_{AB_1})+hg^\beta(C^2_{AB_2}+\ldots +C^2_{AB_{N-1}})\\
&\ \ \ \ +\frac{t}{2}g^{\sqrt{2}}(C^2_{AB_2}+\ldots +C^2_{AB_{N-1}})\Big[g^{\beta-{\sqrt{2}}}(C^2_{AB_1})-g^{\beta-{\sqrt{2}}}(C^2_{AB_2}+\ldots +C^2_{AB_{N-1}})\Big]\\
&\ \ \ \ +\frac{(t-1)^2}{4}g^{2\sqrt{2}}(C^2_{AB_2}+\ldots +C^2_{AB_{N-1}})\Big[g^{\beta-2\sqrt{2}}(C^2_{AB_1})+g^{\beta-2\sqrt{2}}(C^2_{AB_2}+\ldots +C^2_{AB_{N-1}})\\
&\ \ \ \ -2g^{\beta-\sqrt{2}}(C^2_{AB_2}+\ldots +C^2_{AB_{N-1}})g^{-\sqrt{2}}(C^2_{AB_1})\Big]\\
&\ \ \geq g^\beta(C^2_{AB_1})+hg^\beta(C^2_{AB_2}+\ldots+C^2_{AB_{N-1}})\\
&\ \  \ \ +\frac{t}{2}\Big[g^{\sqrt{2}}(C^2_{AB_2})+\ldots+g^{\sqrt{2}}(C^2_{AB_{N-1}})\Big]\cdot\Big[g^{\beta-{\sqrt{2}}}(C^2_{AB_1})-g^{\beta-{\sqrt{2}}}(C^2_{A|B_2\cdots B_{N-1}})\Big]\\
&\ \ \ \ +\frac{(t-1)^2}{4}\Big[g^{2\sqrt{2}}(C^2_{AB_2})+\ldots+g^{2\sqrt{2}}(C^2_{AB_{N-1}})\Big]\cdot
\Big[g^{\beta-{2\sqrt{2}}}(C^2_{AB_1})+\ldots+g^{\beta-{2\sqrt{2}}}(C^2_{AB_{N-1}})\\
&\ \ \ \ -2g^{\beta-\sqrt{2}}(C^2_{A|B_2\cdots B_{N-1}})g^{-\sqrt{2}}(C^2_{AB_1})\Big]\\
&\ \ \geq g^\beta(C^2_{AB_1})+hg^\beta(C^2_{AB_2})+\ldots+h^{N-2}g^\beta(C^2_{AB_{N-1}})\\
&\ \ \ \ +h^{N-3}\cdot\frac{t}{2}\cdot g^{\sqrt{2}}(C^2_{AB_{N-1}})\Big[g^{\beta-{\sqrt{2}}}(C^2_{AB_{N-2}})-g^{\beta-{\sqrt{2}}}(C^2_{AB_{N-1}})\Big]+\ldots\\
&\ \ \ \ +\frac{t}{2}\Big[g^{\sqrt{2}}(C^2_{AB_2})+\ldots+g^{\sqrt{2}}(C^2_{AB_{N-1}})\Big]\cdot\Big[g^{\beta-{\sqrt{2}}}(C^2_{AB_1})-g^{\beta-{\sqrt{2}}}(C^2_{A|B_2\cdots B_{N-1}})\Big]\\
&\ \ \ \ +h^{N-3}\cdot\frac{(t-1)^2}{4}\cdot g^{2\sqrt{2}}(C^2_{AB_{N-1}})\Big[g^{\beta-{2\sqrt{2}}}(C^2_{AB_{N-2}})+g^{\beta-{2\sqrt{2}}}(C^2_{AB_{N-1}})\\
&\ \ \ \ -2g^{\beta-{\sqrt{2}}}(C^2_{AB_{N-1}})g^{-\sqrt{2}}(C^2_{AB_{N-2}})\Big]+\ldots\\
&\ \ \ \ +\frac{(t-1)^2}{4}\Big[g^{2\sqrt{2}}(C^2_{AB_2})+\ldots+g^{2\sqrt{2}}(C^2_{AB_{N-1}})\Big]\cdot
\Big[g^{\beta-{2\sqrt{2}}}(C^2_{AB_1})+\ldots+g^{\beta-{2\sqrt{2}}}(C^2_{AB_{N-1}})\\
&\ \ \ \ -2g^{\beta-{\sqrt{2}}}(C^2_{A|B_2\cdots B_{N-1}})g^{-\sqrt{2}}(C^2_{AB_1})\Big],\\
\end{array}
\end{equation}
where we have utilized the Eq. (\ref{Con1}) and the  monotonicity of $g(x)$ to obtain the first inequality, the third and the forth inequalities are due to the Eq. (\ref{ine for g}) and the monotonicity of the function $g(x)$.

According to Eq. (\ref{thm3e1}) and the fact that $g(\mathcal{C}^2(\rho))=E(\rho)$ for arbitrary two-qubit mixed state $\rho$, we obtain Eq. (\ref{thm3-e}).
\end{proof}

\begin{example}\label{ex2}
 Let us consider the state in (\ref{Con6}) given in Example \ref{ex1}. Set $\lambda_0=\frac{\sqrt{6}}{3},\lambda_1=0,\lambda_2=\frac{\sqrt{2}}{3},\lambda_3=\frac{1}{3},\lambda_4=0$, we have $E_{A|BC}=0.91829$, $E_{AB}=0.68193$, $E_{AC}=0.40416$.
Then $E_{A|BC}^\beta=(0.91829)^\beta\geq E_{AB}^\beta+hE_{AC}^\beta+\frac{\beta}{2\sqrt{2}}E_{AC}^{\sqrt{2}}(E_{AB}^{\beta-{\sqrt{2}}}-E_{AC}^{\beta-{\sqrt{2}}})+
\frac{(\beta-\sqrt{2})^2}{8}E_{AC}^{2\sqrt{2}}(E_{AB}^{\beta-2\sqrt{2}}+E_{AC}^{\beta-2\sqrt{2}}-2E_{AC}^{\beta-\sqrt{2}}E_{AB}^{-\sqrt{2}})
=(0.68193)^\beta+h(0.40416)^\beta\\
+\frac{\beta}{2\sqrt{2}}(0.40416)^{\sqrt{2}}\Big[(0.68193)^{\beta-{\sqrt{2}}}-(0.40416)^{\beta-{\sqrt{2}}}\Big]+\frac{(\beta-\sqrt{2})^2}{8}(0.40416)^{2\sqrt{2}}\Big[(0.68193)^{\beta-2\sqrt{2}}+(0.40416)^{\beta-2\sqrt{2}}-2 (0.40416)^{\beta-\sqrt{2}}(0.68193)^{-\sqrt{2}}\Big]$. While   the result in \cite{JZSZ} is $E_{AB}^\beta+hE_{AC}^\beta+\frac{\beta}{2\sqrt{2}}E_{AC}^{\sqrt{2}}(E_{AB}^{\beta-{\sqrt{2}}}-E_{AC}^{\beta-{\sqrt{2}}})=(0.68193)^\beta+h(0.40416)^\beta+\frac{\beta}{2\sqrt{2}}(0.40416)^{\sqrt{2}}\Big[(0.68193)^{\beta-{\sqrt{2}}}-(0.40416)^{\beta-{\sqrt{2}}}\Big] $. We can see that our results are better than the ones in \cite{JZSZ}, see FIG. \ref{2}.

\begin{figure}
\centering
\includegraphics[width=0.75\textwidth]{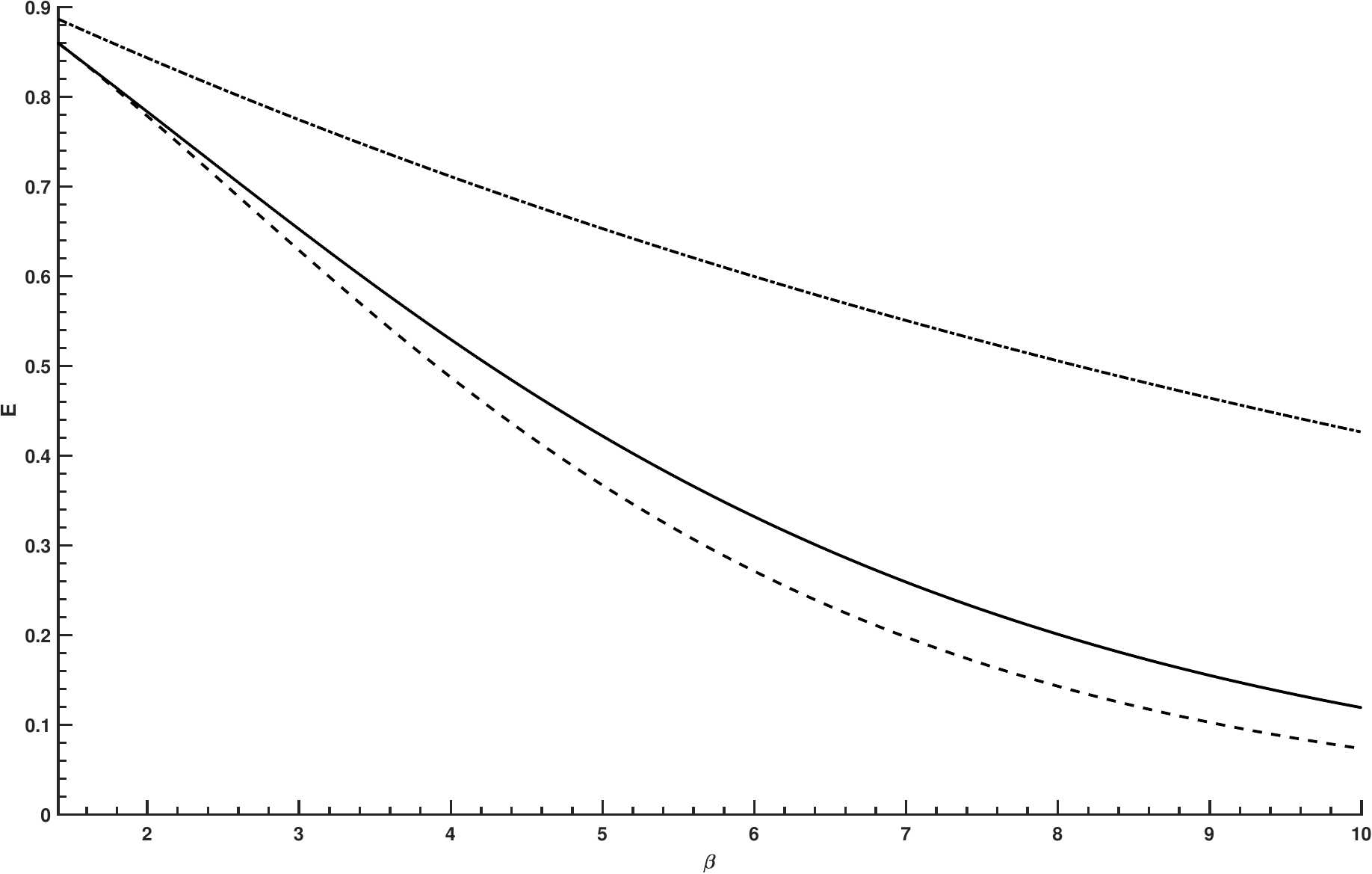}
\caption{Dash dotted line ,   $E^\beta_{A|BC}$ as a function of $\beta$ ($\sqrt{2}\leq\beta\leq10$); solid line,  the lower bound of $E^\beta_{A|BC}$ as a function of $\beta$ ($\sqrt{2}\leq\beta\leq10$) in Eq. (\ref{thm3-e});  dash  line, the lower bound of $E^\beta_{A|BC}$ as a function of $\beta$ ($\sqrt{2}\leq\beta\leq10$) in \cite{JZSZ}.}
\label{2}
\end{figure}
\end{example}

\section{Tighter Monogamy Relations using Negativity}\label{sec4}
The negativity is a well-known quantifier of bipartite entanglement. Given a bipartite state $\rho_{AB}$ in Hilbert space $H_A\otimes H_B$, the negativity is defined as \cite{Vidal}
\begin{eqnarray}
\N(\rho_{AB})=\frac{\|\rho_{AB}^{T_A}\|-1}{2},
\end{eqnarray}
where $\rho_{AB}^{T_A}$ is the partial transposed matrix of $\rho_{AB}$ with respect to the subsystem $A$
and $\|X\|$ denotes the trace norm of $X$, i.e. $\|X\|={\rm Tr}\sqrt{XX^{\dag}}$.
For convenience, we use the definition of negativity as $\|\rho_{AB}^{T_A}\|-1$ \cite{Kim}.

If a bipartite pure state
$|\phi\rangle_{AB}$ with the Schmidt decomposition,
$|\phi\rangle_{AB} = \sum_{i}\sqrt{\lambda_i}|ii\rangle$, $\lambda_i\geq 0$, $\sum_{i}\lambda_i = 1$,
then \cite{Kim}
\begin{equation}\label{N}
\mathcal{N}(|\phi\rangle_{AB}) = 2\sum_{i< j}\sqrt{\lambda_i\lambda_j}.
\end{equation}
From the definition of concurrence (\ref{def of C}), we have
\begin{equation}
C(|\phi\rangle_{AB}) = 2\sqrt{\sum_{i< j}\lambda_i\lambda_j}.
\end{equation}
As a consequence, for any bipartite pure state $|\phi\rangle_{AB}$ with Schmidt rank $2$, one has $\mathcal{N}(|\phi\rangle_{AB})=C(|\phi\rangle_{AB})$.

For a mixed state $\rho_{AB}$, the convex-roof extended negativity (CREN) is given by
\begin{equation}
\N_c(\rho_{AB})=\min\limits_{\{p_i,|\phi_i\rangle\}}\sum_{i}p_i\N(|\phi_i\rangle),
\end{equation}
where the minimum is taken over all possible pure state decomposition of $\rho_{AB}$. CREN gives a perfect discrimination between PPT bound entangled states and separable states in any bipartite quantum system \cite{Lee}.
It follows that for any $2\otimes d ~(d\geq 2)$ mixed state $\rho_{AB}$, we have
\begin{equation}
\N_c(\rho_{AB})=\min\limits_{\{p_i,|\phi_i\rangle\}}\sum_{i}p_i\N(|\phi_i\rangle)=\min\limits_{\{p_i,|\phi_i\rangle\}}\sum_{i}p_iC(|\phi_i\rangle)=C(\rho_{AB}).
\end{equation}

According to the relation between CREN and concurrence, we have the following results for the lower bound of $\N_{c{A|B_1\cdots B_{N-1}}}^{\beta}$.

\begin{theorem}\label{thm4}
For any $N$-qubit mixed state $\rho_{AB_1\cdots B_{N-1}}$, if $\N_{c{AB_i}}\geq \N_{c{AB_{i+1}\cdots B_{N-1}}}$ for $i=1,2,\ldots,m$, and $\N_{c{AB_j}}\leq \N_{c{A|B_{j+1}\cdots B_{N-1}}}$ for $j=m+1,\ldots,N-2$, $\forall ~ 1\leq m\leq N-3$, then we have
\begin{equation}
\begin{array}{rl}
\N_{c{A|B_1\cdots B_{N-1}}}^\beta\geq\sum\limits_{i=1}^{m}h^{i-1}(\N_{c{AB_i}}^\beta+R_{AB_i})+
h^m\sum\limits_{j=m+1}^{N-2}(h\N_{c{AB_j}}^\beta+R_{AB_j}^1)+h^m \N_{c{AB_{N-1}}}^\beta,
\end{array}
\end{equation}
for all $N\geq4$,  $\beta\geq2$, where $h=2^{\frac{\beta}{2}}-1$, $R_{AB_i}=\frac{\beta}{4}\N_{c{A|B_{i+1}\ldots B_{N-1}}}^2(\N_{c{AB_i}}^{\beta-2}-\N_{c{A|B_{i+1}\ldots B_{N-1}}}^{\beta-2})+
\frac{(\beta-2)^2}{16}\N_{c{A|B_{i+1}\ldots B_{N-1}}}^4(\N_{c{AB_i}}^{\beta-4}+\N_{c{A|B_{i+1}}\ldots B_{N-1}}^{\beta-4}-2\N_{c{A|B_{i+1}\ldots B_{N-1}}}^{\beta-2}\N_{c{AB_i}}^{-2})$,
$R_{AB_j}^1=\frac{\beta}{4}\N_{c{AB_j}}^2\\
(\N_{c{A|B_{j+1}\ldots B_{N-1}}}^{\beta-2}-\N_{c{AB_j}}^{\beta-2})+\frac{(\beta-2)^2}{16}\N_{c{AB_j}}^4(\N_{c{A|B_{j+1}\ldots B_{N-1}}}^{\beta-4}+\N_{c{AB_j}}^{\beta-4}-2\N_{c{AB_j}}^{\beta-2}\N_{c{A|B_{j+1}\ldots B_{N-1}}}^{-2})$.
\end{theorem}

\begin{theorem}\label{thm5}
For any $N$-qubit mixed state $\rho_{AB_1\cdots B_{N-1}}$, if $\N_{c{AB_i}}\geq \N_{c{A|B_{i+1}\cdots B_{N-1}}}$ for $i=1,2,\ldots,N-2$, then we can obtain
\begin{equation}\label{thm5-e}
\begin{array}{rl}
\N_{c{A|B_1\cdots B_{N-1}}}^\beta\geq\sum\limits_{i=1}^{N-2}h^{i-1}(\N_{c{AB_i}}^\beta+R_{AB_i})+h^{N-2} \N_{c{AB_{N-1}}}^\beta,
\end{array}
\end{equation}
for all $N\geq3$,  $\beta\geq2$, where $h=2^{\frac{\beta}{2}}-1$, $R_{AB_i}=\frac{\beta}{4}\N_{c{A|B_{i+1}\ldots B_{N-1}}}^2(\N_{c{AB_i}}^{\beta-2}-\N_{c{A|B_{i+1}\ldots B_{N-1}}}^{\beta-2})+
\frac{(\beta-2)^2}{16}\N_{c{A|B_{i+1}\ldots B_{N-1}}}^4(N_{c{AB_i}}^{\beta-4}+\N_{c{A|B_{i+1}\ldots B_{N-1}}}^{\beta-4}-2\N_{c{A|B_{i+1}\ldots B_{N-1}}}^{\beta-2}\N_{c{AB_i}}^{-2})$.
\end{theorem}

\begin{example} \label{ex3}
Let us consider the state in (\ref{Con6}) given in Example \ref{ex1}. We have $\N_{cA|BC}=2\lambda_0\sqrt{\lambda_2^2+\lambda_3^2+\lambda_4^2}$, $\N_{cAB}=2\lambda_0\lambda_2$ and
$\N_{cAC}=2\lambda_0\lambda_3$. Set $\lambda_0=\frac{\sqrt{2}}{3},\lambda_1=0,\lambda_2=\frac{\sqrt{5}}{3},\lambda_3=\frac{\sqrt{2}}{3},\lambda_4=0$.
 We have $\N_{c{A|BC}}^\beta\geq \N_{c{AB}}^\beta+h\N_{c{AC}}^\beta+\frac{\beta}{4}\N_{c{AC}}^2(\N_{c{AB}}^{\beta-2}-\N_{c{AC}}^{\beta-2})+\frac{(\beta-2)^2}{16}\N_{c{AC}}^4(\N_{c{AB}}^{\beta-4}+\N_{c{AC}}^{\beta-4}
-2\N_{c{AC}}^{\beta-2}\N_{c{AB}}^{-2})
=(\frac{2\sqrt{10}}{9})^\beta+h(\frac{4}{9})^\beta+\frac{\beta}{4}(\frac{4}{9})^2\Big[(\frac{2\sqrt{10}}{9})^{\beta-2}-(\frac{4}{9})^{\beta-2}\Big]+\frac{(\beta-2)^2}{16}(\frac{4}{9})^4\Big[(\frac{2\sqrt{10}}{9})^{\beta-4}+(\frac{4}{9})^{\beta-4}-2(\frac{4}{9})^{\beta-2}(\frac{2\sqrt{10}}{9})^{-2})$.
While  the result in\cite{JZSZ} is $\N_{cAB}^\beta+h\N_{cAC}^\beta+\frac{\beta}{4}\N_{cAC}^2(\N_{cAB}^{\beta-2}-\N_{cAC}^{\beta-2})=(\frac{2\sqrt{10}}{9})^\beta+h(\frac{4}{9})^\beta+\frac{\beta}{4}(\frac{4}{9})^2\Big[(\frac{2\sqrt{10}}{9})^{\beta-2}-(\frac{4}{9})^{\beta-2}\Big]$. We can see that our result is better than the one in \cite{JZSZ} for $\beta\geq2$, see FIG. \ref{3}.

\begin{figure}
\centering
\includegraphics[width=0.75\textwidth]{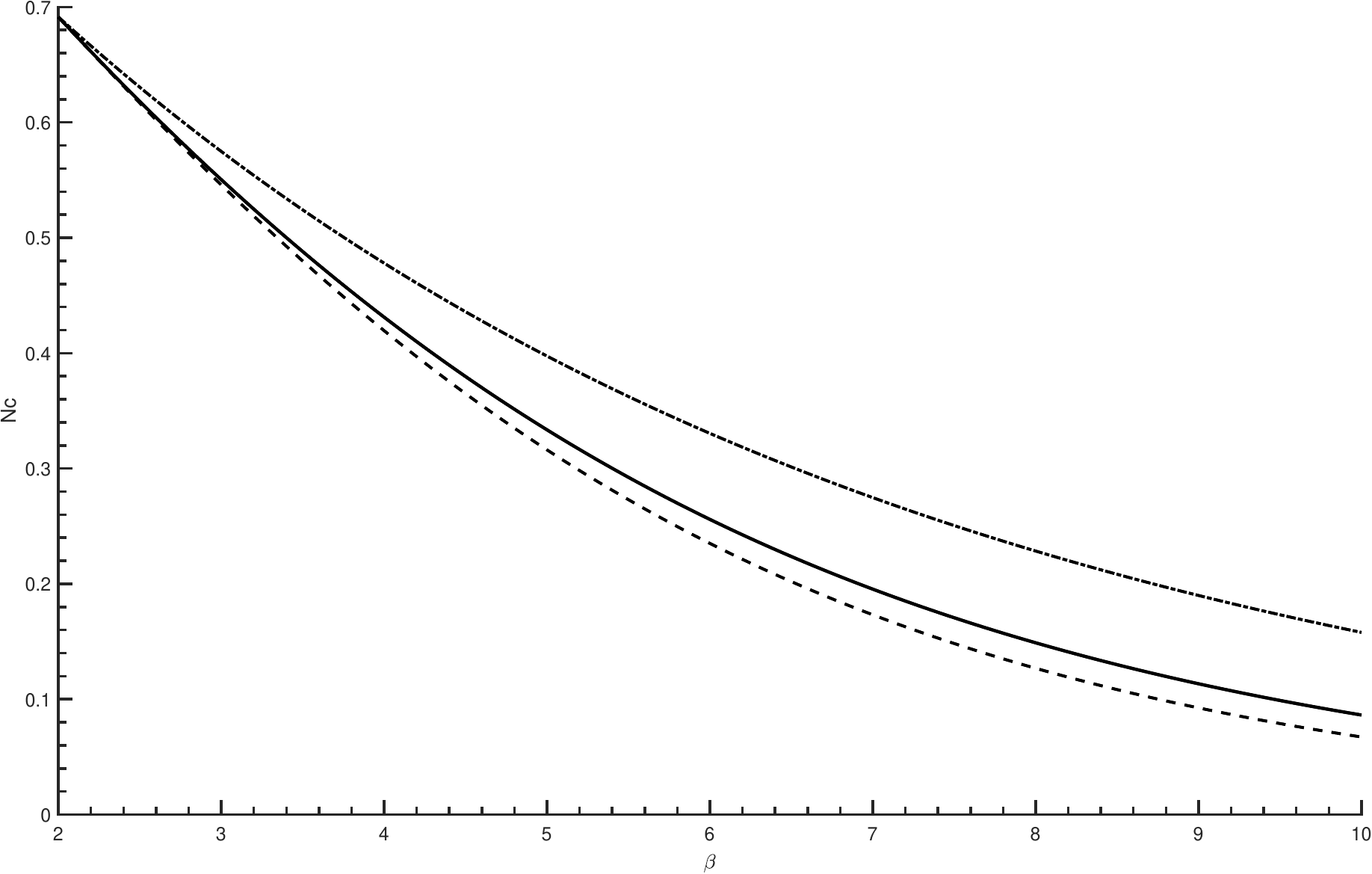}
\caption{Dash dotted line ,   $\N^\beta_{cA|BC}$ as a function of $\beta$ ($2\leq\beta\leq10$); solid line,  the lower bound of $\N^\beta_{cA|BC}$ as a function of $\beta$ ($2\leq\beta\leq10$) in Eq. (\ref{thm5-e});  dash  line, the lower bound of $\N^\beta_{cA|BC}$ as a function of $\beta$ ($2\leq\beta\leq10$) in \cite{JZSZ}.}
\label{3}
\end{figure}
\end{example}

\section{Conclusion}\label{sec5}
Entanglement monogamy  relations are fundamental properties of multipartite entangled states. In this paper, we  have provided
 the multipartite entanglement based on the monogamy relations for $\beta$th power of concurrence $C_{A|B_1\ldots B_{N-1}}^{\beta}$ ($\beta\geq 2$),  entanglement of formation $E_{A|B_1\ldots B_{N-1}}^{\beta}$ ($\beta\geq \sqrt{2}$) and convex-roof extended negativity $\N_{cA|B_1\ldots B_{N-1}}^{\beta}$ ($\beta\geq 2$).
Our monogamy relations have larger lower bounds and are tighter than the  existing results \cite{JZSZ}.
 These tighter monogamy inequalities can also provide a finer description of the entanglement distribution. In  multi-qubit system, our research results provide a rich reference for future research on multi-party quantum entanglement. Our method can also be applied to the study of other properties of monogamy related to quantum correlations.

\section*{Acknowledgments}
This work is  supported by the Yunnan Provincial Research Foundation for Basic Research, China (Grant No. 202001AU070041), the Research Foundation of Education Bureau of Yunnan Province, China (Grant No. 2021J0054), the Basic and Applied Basic Research Funding Program of Guangdong Province (Grant No. 2019A1515111097), the Natural Science Foundation of Kunming University of Science and Technology (Grant No. KKZ3202007036, KKZ3202007049).

\section*{Data Availability Statement}

All data generated or analysed during this study are included in this submitted article.

\end{document}